\documentclass[11pt]{article}

\PassOptionsToPackage{svgnames,table}{xcolor}
\usepackage[svgnames,table]{xcolor}
\usepackage{subcaption}
\usepackage{tikz}
\usetikzlibrary{arrows.meta,positioning,calc}
\usepackage[T1]{fontenc}
\usepackage[utf8]{inputenc}
\usepackage{enumitem}
\usepackage{graphicx}
\usepackage{fancybox}
\usepackage{comment}
\usepackage{nameref}
\usepackage{bbm}

\definecolor{ForestGreen}{rgb}{0.1333,0.5451,0.1333}
\definecolor{DarkRed}{rgb}{0.8,0,0}
\definecolor{Red}{rgb}{1,0,0}
\usepackage[linktocpage=true,
pagebackref=true,colorlinks,
linkcolor=ForestGreen,citecolor=ForestGreen,
bookmarks,bookmarksopen,bookmarksnumbered]
{hyperref}
\usepackage[nottoc,numbib]{tocbibind}

\usepackage{amsmath}
\usepackage{amssymb}
\usepackage{amsthm}

\usepackage[ruled, noend, linesnumbered]{algorithm2e}
\SetKwComment{Comment}{/* }{ */} \SetKwInOut{Input}{Input} 
\SetKwInOut{Output}{Output}

\usepackage{float}
\usepackage{cleveref}

\newtheorem{theorem}{Theorem}[section]

\newtheorem{corollary}[theorem]{Corollary}
\newtheorem{lemma}[theorem]{Lemma}

\newtheorem{claim}[theorem]{Claim}

\newtheorem{fact}[theorem]{Fact}

\newtheorem{definition}[theorem]{Definition}
\newtheorem{remark}[theorem]{Remark}

\newtheorem*{theorem*}{Theorem}
\newtheorem*{corollary*}{Corollary}
\newtheorem*{conjecture*}{Conjecture}
\newtheorem*{lemma*}{Lemma}
\newtheorem*{thm*}{Theorem}
\newtheorem*{prop*}{Proposition}
\newtheorem*{obs*}{Observation}
\newtheorem*{definition*}{Definition}

\newtheorem*{remark*}{Remark}
\newtheorem*{rec*}{Recommendation}

\newenvironment{fminipage}%
  {\begin{Sbox}\begin{minipage}}%
  {\end{minipage}\end{Sbox}\fbox{\TheSbox}}

\def\defeq{\stackrel{\mathrm{def}}{=}}

\newcommand{\veryshortarrow}[1][3pt]{\mathrel{%
   \hbox{\rule[\dimexpr\fontdimen22\textfont2-.2pt\relax]{#1}{.4pt}}%
   \mkern-4mu\hbox{\usefont{U}{lasy}{m}{n}\symbol{41}}}}

\def\norm#1{\left\| #1 \right\|}

\def\rec#1{\left\| #1 \right\|_{\veryshortarrow}}

\newcommand\DDelta{\boldsymbol{\mathit{\Delta}}}

\def\aa{\pmb{\mathit{a}}}
\renewcommand\gg{\boldsymbol{\mathit{g}}}
\newcommand\bb{\boldsymbol{\mathit{b}}}
\newcommand\cc{\boldsymbol{\mathit{c}}}
\newcommand\dd{\boldsymbol{\mathit{d}}}

\newcommand\ff{\boldsymbol{\mathit{f}}}

\renewcommand\ll{\boldsymbol{\mathit{l}}}
\newcommand\pp{\boldsymbol{\mathit{p}}}
\newcommand\qq{\boldsymbol{\mathit{q}}}

\newcommand\uu{\boldsymbol{\mathit{u}}}

\newcommand\veczero{\boldsymbol{0}}
\newcommand\vecone{\boldsymbol{1}}

\renewcommand\AA{\boldsymbol{\mathit{A}}}
\newcommand\BB{\boldsymbol{\mathit{B}}}

\newcommand\LL{\boldsymbol{\mathit{L}}}

\newcommand\R{\mathbb{R}}

\DeclareMathOperator*{\diag}{diag}

\newcommand{\concat}{\oplus}

\renewcommand{\root}{\mathsf{root}}

\newcommand{\Enc}{\textsc{Enc}}

\renewcommand{\hat}{\widehat}
\renewcommand{\tilde}{\widetilde}

\DeclareFontFamily{U}{mathb}{\hyphenchar\font45}
\DeclareFontShape{U}{mathb}{m}{n}{<5> <6> <7> <8> <9> <10> gen * mathb
<10.95> mathb10 <12> <14.4> <17.28> <20.74> <24.88> mathb12}{}
\DeclareSymbolFont{mathb}{U}{mathb}{m}{n}
\DeclareMathSymbol{\rcirclearrow}{\mathbin}{mathb}{'367}

\newif\ifrandom
\randomtrue

\newcommand{\poly}{{\mathrm{poly}}}

\newcommand{\str}{{\mathsf{str}}}

\newcommand{\todolater}[1]{}

\DeclareUnicodeCharacter{2113}{$\ell$}

\renewcommand{\root}{\mathsf{root}}

\usepackage{fullpage}
\usepackage[nottoc,numbib]{tocbibind}
\usepackage{amsthm}

\title{A Simpler and Faster Min-Cost Flow Solver \\ via Min-Ratio Cycles from Distance Oracles}
\author{
Rasmus Kyng \\
\and
Simon Meierhans \\
\and
Maximilian Probst Gutenberg \\
\and 
Aurelio Sulser \\
}
\date{}

\DeclareMathOperator{\pproj}{proj}

\DeclareMathOperator{\vcong}{vcong}

\DeclareMathOperator{\dist}{dist}

\begin{document}

\maketitle

\begin{abstract}
The first almost-linear time maximum and minimum cost flow algorithm of Chen-Kyng-Liu-Peng-Probst Gutenberg-Sachdeva \cite{maxflow}, reduced these flow objectives to a sequence of min-ratio cycle problems.
Solving this core primitive requires approximately minimizing the ratio of a linear gradient term and an undirected length term. In \cite{maxflow} and the subsequent work of Chen-Kyng-Liu-Meierhans-Probst Gutenberg \cite{incrflow}, intricate data structures were given to solve the min-ratio problem.

We show that such a cycle can be extracted directly from the dynamic distance oracle of Kyng-Meierhans-Probst Gutenberg \cite{KMPG24} using linearity. This simplifies previous algorithms that relied on multiple additional steps to extract the cycle, and can be seen as evidence that solving the min-ratio cycle problem really is all about distances. As a result, we obtain a faster primal maxflow and min-cost flow solver that also extends to incremental graphs.
\end{abstract}

\tableofcontents

\pagebreak

\section{Introduction}

The maximum flow problem and its generalization, the min-cost flow problem, are among the most fundamental problems in algorithmic graph theory. Since the dawn of the formal study of graph algorithms, they have been intensively studied \cite{FF56, D70, ET75, ST83, T85, GT89, DS08, GT14, M13, M16, CKMST11, S13, S17, BGS21, van2021minimum, van2022faster, maxflow, detmax, incrflow, decrflow, bernstein2024maximum, chuzhoy2024faster, haeupler2025parallel,  bernstein2025combinatorial, haeupler2026dag}, have served as a test bed to develop some of the most well-known and popular algorithmic paradigms, and serve as the key primitive in many related algorithmic problems \cite{li2025deterministic, li2021vertex, he2024cactus, abboud2023all, ding2024deterministic, abboud2025deterministic}. Formally, the problem is defined as follows.

\begin{definition}[Min-cost flow]\label{def:mincostFlow}
Given a directed graph $G=(V,E)$ with corresponding edge-vertex incidence matrix $\BB$, with edge capacities \(\uu\) and costs \(\cc\) and a demand vector
\(\dd \perp \vecone\), the min-cost flow problem seeks to find
\[
\ff^\star
=
\arg\min_{\substack{\BB^\top \ff=\dd \\ \veczero \leq \ff \leq \uu}} \cc^\top \ff .
\]
In the thresholded min-cost flow problem, we are additionally given a fixed threshold $F \in \mathbf{R}$ and the algorithm has to decide if $\cc^\top \ff^* \leq F$, i.e. if the min-cost flow solution has cost at most $F$.
\end{definition}

\paragraph{Fast Algorithms for Min-Cost Flow.} While an almost-linear time algorithm for the max-flow and min-cost flow problem remained elusive until recently, the work of Chen et al. \cite{maxflow} demonstrated the existence of such an $m^{1+o(1)}$ time algorithm. This breakthrough was the culmination of many years of advances in developing and integrating convex optimization, graph-algorithmic techniques, and data structures for dynamic problems.
\cite{maxflow} suggested a novel $\ell_1$-interior point method ($\ell_1$-IPM) that reduced the min-cost flow problem to solving $m^{1+o(1)}$ approximate min-ratio cycle problems as defined below. 

\begin{definition}[Min-Ratio Cycle]\label{def:minRatioCycle}
Given graph $G=(V,E)$ with lengths $\ll \in \mathbb{R}^E_{\geq 0}$ and gradients $\gg \in \mathbb{R}^E$, we denote for any circulation $\BB \cc = \veczero$, the \emph{ratio} achieved by $r(\cc) =  \frac{\gg^\top \cc}{\|\LL \cc\|_1}$ where $\LL = \diag(\ll)$. We let the min-ratio $r^* = \min_{\BB \cc = 0} r(\cc)$ and say a circulation $\cc$ is a min-ratio cycle if $r(\cc) = r^*$. We say a value $\tilde{r}$ is $\alpha$-approximate min-ratio if $\tilde{r} \in [r^*, r^* / \alpha]$ for $\alpha \geq 1$. We say a circulation $\cc$ is an $\alpha$-approximate min-ratio cycle if $r(\cc)$ is $\alpha$-approximate min-ratio.
\end{definition}
\begin{remark}
We remark that there always is a simple cycle $\cc$ that realizes the min-ratio.
\end{remark}

\cite{maxflow} complemented their $\ell_1$-IPM with a new graph data structure that solves the sequence of approximate min-ratio cycle problems in $m^{1+o(1)}$ total time by exploiting the stability of the update sequence. More precisely, they show that any two consecutive min-ratio cycle problems only differ in $m^{o(1)}$ coordinates with respect to $\uu$ and $\gg$, on average.

Unfortunately, the data structure developed in \cite{maxflow} for the dynamic min-ratio problem, suffers from working only for special, highly-structured update sequences. While Chen et al. where able to show that their data structure indeed works for updates generated by the $\ell_1$-IPM, the resulting proof, which is already quite intricate, intertwines their data structure analysis with guarantees on the well-behavedness of the IPM updates. This makes for a dense and hard-to-read proof. A de-randomization was achieved in \cite{detmax}, however, at the cost of a further complication of the algorithm, and without resolving this entanglement.

\paragraph{Current State-of-the-Art for Minimum-Cost Flow.} In \cite{incrflow}, a novel approach was given to solve the dynamic min-ratio problem efficiently for \textbf{any} adversarial update sequence. This resulted in a much more modular proof for a min-cost flow algorithm, separating the $\ell_1$-IPM analysis and data structure analysis. Further, as was shown by \cite{BLS23}, solving dynamic min-ratio cycle efficiently against any adversary directly yields an algorithm for min-cost flows even for \emph{incremental graphs}, i.e. input graphs that undergo edge insertions, one-by-one.
Thus  \cite{incrflow} gave the first almost-linear time algorithm for min-cost flow in incremental graphs.

Further, \cite{incrflow} then served as a template for \cite{decrflow} where a superset of the authors develop a data structure for the subproblem obtained by running the $\ell_1$-IPM on the dual problem: the min-ratio cut problem. This results in a static min-cost flow algorithm that is significantly simpler than the previous approaches, and this approach extends to \emph{decremental graphs}, i.e. input graphs that undergo edge deletions, one-by-one.
The data structure for min-ratio cuts relies on dynamically maintaining a tree-cut sparsifier of the graph, and from this, extracting approximately optimal cuts.







\paragraph{Min-Ratio Cycles Against Oblivious Adversaries.}
In the original almost-linear time min-cost flow algorithm of \cite{maxflow}, approximate min-ratio cycles are extracted from a dynamic, recursive graph sparsifier hierarchy that alternates between edge sparsification using spanners and vertex sparsification using  `distance-approximation core graphs'.

Their core graph achieves two things at once: (1) it partitions the graph vertex set into relatively few, small sets (2) it ensures that an approximate min-ratio cycle is either contained in a partition set, or survives after contracting each set into a single vertex. Overall, the vertex count is substantially reduced, and either a good cycle can be found by checking each partition set, or a good cycle exists in the resulting smaller graph.
Unfortunately, this procedure only works against oblivious adversaries.\footnote{The authors strengthen the result to work against a slightly stronger adversary in the presence of `hints' that can be obtained from the IPM. This is reminiscent of the literature on algorithms with predictions. }
We will refer to this as a \emph{vertex sparsification dichotomy} against oblivious adversaries.
To use the dichotomy, the authors of \cite{maxflow} develop a sophisticated analysis framework that allows them to treat the update sequence of a static min-cost flow $\ell_1$-IPM nearly as though it was generated by an oblivious adversary.

Their edge sparsification procedure relied on a novel dynamic spanner, which also exhibited dichotomy: either the edge-sparsified graph contains an approximate min-ratio cycle, or the sparsification identifies a good cycle. 
This \emph{edge sparsification dichotomy} works against adaptive adversaries.

\paragraph{Min-Ratio Cycles Against All Adversaries.}
In \cite{incrflow}, constructed a deterministic min-ratio cycle data structure with almost-linear work for a linear number of updates, thus solving the problem against all adversaries, and solving incremental minimum-cost flow in almost linear time via the reduction of \cite{BLS23}.
The algorithm of \cite{incrflow} uses a powerful fully-dynamic APSP data structure from \cite{KMPG24}. 
Before we review the \cite{incrflow} min-ratio data structure, we will first briefly discuss this APSP data structure.
At a high level, this data structure is similar to the original dynamic sparsifier hierarchy of \cite{maxflow}, but replaces their oblivious-adversary vertex sparsifier with a new vertex sparsifier that differs in two crucial ways: Firstly, this sparsifier replaces the goal of maintaining min-ratio cycles with a weaker goal of only maintaining pairwise distances between terminals (remaining vertices). 
Secondly, the authors develop a much stronger, deterministic framework for preserving these distances.

As a side-effect, the \cite{KMPG24} APSP data structure also yields a small collection of dynamic spanning trees of the input graph, such that the pairwise distance between any two vertices in original graph is approximately preserved in at least one tree.
This dynamic tree collection is the starting point for the deterministic min-ratio cycle data-structure of \cite{incrflow}.
Using these trees, the authors dynamically maintain hierarchy of low-diameter decompositions, and from these decompositions they construct a dynamic version of $\ell_1$-oblivious routing of \cite{RGHZL22}.
Next, they show how to extract \emph{another} small collection of trees from this routing, such that an approximately optimal min-ratio cycle is generated as a fundamental cycle of one of these trees, i.e., a cycle consisting of one off-tree edge and the tree path between its endpoints.
Unfortunately, there is no easy way to identify \emph{which} edge is good.
This is because the trees are changing dynamically, and each tree update may change a huge number of fundamental cycles.
However, the trees are still useful: 
the authors use these new trees as the starting point for building a new vertex sparsifying core graph that deterministically preserves min-ratio cycles, or locally identify a cycle. 
Thus, they recover the vertex sparsification dichotomy of the original \cite{maxflow} vertex sparsifier, but now against any adversary.
Unfortunately, this vertex sparsifier is more unstable that the prior approach, which in turn necessitates the development of a new dynamic spanner that can compensate for this instability by handling a large number of edge insertions.
The net result is a 
modular separation of min-ratio data structure and $\ell_1$-IPM, but this is achieved at the expense of immense complication of the data structure.

\subsection{Our Contribution}

We develop a much simpler deterministic min-ratio cycle data structure, yielding a faster and simpler incremental min-cost flow algorithm. 
We replace the entire apparatus of \cite{incrflow} 
with a method that directly extracts
min-ratio cycles from the \cite{KMPG24} fully-dynamic APSP data structure.

To understand this, we have to unpack the \cite{KMPG24} vertex sparsifer slightly. 
Given an original graph, and subset $T$ of the vertices, again called \emph{terminals}, goal of this data structure is to produce a smaller graph that preserves pairwise distance between the vertices of $T$.
Concretely, they output a graph whose vertex set is a superset of $T$ with size $|T|^{1+o(1)}$, and ensure that this graph undergoes few changes as the input graph changes.
The first key ingredient of this sparsifier is the construction of a dynamic version of so-called `ASZ-paths' \cite{ASZ}:
this is a relatively collection of paths in the original graph, such that preserving all of their lengths guarantees preserving the distances between vertices in $T$ up to a constant factor. 

Preserving these paths would suffice, but unfortunately, they change too rapidly as the input graph changes.
In the language of dynamic graph algorithms, we say their \emph{recourse} is too high.
To circumvent this, \cite{KMPG24} introduces their second key ingredient.
Remarkably, they show how to construct an entire set of possible future ASZ paths and construct a collection of low-stretch trees that preserve these with low recourse. 
From these trees, a distance-preserving core graph vertex sparsifier can be constructed.

\paragraph{A Simpler Approach to Min-Ratio Cycles.}
Armed with an understanding of \cite{KMPG24}, we can state our core result:
The \cite{KMPG24} vertex sparsifier admits a deterministic vertex sparsification dichotomy for min-ratio cycles!
To make sense of this statement, we have to add gradients (a linear function on the edges) to the graph where we apply the vertex sparsifier, and we have to do some bookkeeping to make sure the recourse of these gradients is not too high.
Our dichotomy is more subtle than earlier approaches, but the central idea can be explained the context of the vertex sparsifier as a static object.
We show how to either identify a good cycle from the sparsifier, or ensure that one exists the the sparsified graph. But, the former case needs to be broken down further: First, we construct an `ASZ-dichotomy' where checking for min-ratio cycles in a ball-cover of the graph either identifies a good cycle, or ensures that one is supported on ASZ-paths.
Second, we construct a `forest dichotomy', where either a good cycle is identified as a fundamental cycle of a tree used in the core-graph construction, or the core graph contains a good cycle. 
Crucially, the latter step only works for cycles supported on ASZ-paths, necessitating the former step.

Equipped with our new vertex sparsifier dichotomy, we can now build a sparsifier hierarchy that alternates between edge and vertex sparsification and directly and deterministically extract min-ratio cycles from this. Combined with an $\ell_1$-IPM, this yields the following result.

\begin{theorem}[Simpler and Faster (Incremental) Min-Cost Flow ]\label{thm:mainIntro}
There is an algorithm that given an (incremental) directed graph $G=(V,E)$ with polynomially bounded edge capacities \(u\) and costs \(c\), a threshold $F$,
a demand vector \(\dd \perp \mathbf{1}\), that for $\chi = e^{O(\log^{20/21} m \log \log m)}$, maintains the solution to the thresholded min-cost flow problem (see \Cref{def:mincostFlow}) using total time $m \chi$. If the min-cost flow has cost at most $F$, the algorithm can also return a feasible flow of cost at most $F$ in additional time $m \chi$.
\end{theorem}

In summary, given the  \cite{KMPG24} distance vertex sparsifier, we construct a very simple min-ratio vertex sparsifier via a new sparsification dichotomy.
A priori, the existing sparsifier was only designed to handle distances, and we believe our method suggests that fully-dynamic APSP data structures can be adapted to solve the min-ratio cycle problem, provided they allow efficient path flow updates.\footnote{The  \cite{KMPG24} differs from other APSP data structures by providing efficient implicit access its paths by representing them in a small number of trees. This makes it possible to efficiently route flow along these paths.} 
This in turn establishes the goal of fully-dynamic APSP with polylogarithmic approximation factor and average-case update time as a promising route toward nearly-linear maximum and minimum-cost flow.

\section{Preliminaries} \label{sec:preliminaires}

\paragraph{Vectors and Basic Facts. } We denote vectors as bold lower case letters $\aa$, and matrices as bold upper case letters $\AA$. For a vector $\aa \in \R^V$ for some set $V$ and a subset $S \subseteq V$, we let $\aa(S) = \sum_{e \in S} \aa(e)$. We use the following two facts throughtout.

\begin{fact}[Dan's Favorite Inequality]
    \label{fact:ratio_inequality}
    For $\aa \in \R^S$ and $\bb \in \R^S_{> 0}$ we have $\min_{i \in S} \frac{\aa(i)}{\bb(i)} \leq \frac{\aa(S)}{\bb(S)}$. 
\end{fact}

\begin{fact}
\label{fac:expapprox}
For $0 \leq x \leq 1/2$,
$\exp(x) \leq 1+2x$
and $\exp(-x) \geq 1-2x$.
\end{fact}

\paragraph{Graphs. } We denote (multi-)graphs as $G = (V, E)$. They often come with additional information relating to edges. We use $\ll \in \R^E_{\geq 0}$ to refer to edge lengths, which relate to distances. We use $\uu \in \R^E_{\geq 0}$ to refer to capacities which define the throughput of cuts. Finally, we sometimes associate gradients $\gg \in \R^E$ with edges. Unlike lengths and capacities, gradients can assume negative values. 

\paragraph{Trees. } A tree $T = (V,E)$ is a connected graph with $|E| = |V| - 1$. For $u, v \in V$, we let $T[u,v]$ denote the unique simple path between $u$ and $v$ in $T$. A rooted tree is a tree with a designated root vertex, and a rooted forest is a collection of rooted trees. 

\paragraph{Distances, Paths and Cycles. } For a graph $G = (V, E, \ll)$, we call a series of edges a path $P = ((v_1, v_2), (v_2, v_3), \ldots (v_{k - 1}, v_k))$ if the head of the next edge is the tail of the previous edge. A path that starts and ends in the same vertex is referred to as a circulation, and if each vertex appears at most twice as an endpoint, we say it is a cycle. In this article, paths can contain the same edge multiple times.\footnote{This is sometimes called a walk in the literature.} For a path $P$, we let $\ll(P) \defeq \sum_{e \in E} \ll(e)$ denote its length. 
We let $\pi_G(u, v)$ denote an arbitrary but fixed shortest path between $u$ and $v$ in $G$. We let $\dist_{G}(u,v)$ denote the length of the shortest path, i.e. $\dist_{G}(u,v) \defeq \ll(\pi(u,v))$.

Even though the graph $G$ is undirected, we associate some arbitrary direction with its edges. Then, we let $\sigma((u,v)) = 1$ and $\sigma((v,u)) = -1$ for an edge with associated direction $(u, v)$. We let $\vecone_{P} = \sum_{e \in P} \sigma(e) \cdot \vecone_e$. 

\begin{fact}[Cycle decomposition of a circulation flow.]
    \label{fac:cycledecomp}
For every graph $G = (V,E)$ with incidence matrix $\BB$ and every circulation $\DDelta \in \R^{E}$, i.e. any vector $\DDelta$ s.t.   $\BB\DDelta = \veczero$, there exists a decomposition $\DDelta = \sum_k \cc_k$ into at most $m$ cycles $\{\cc_k\}_{k = 1}^m$, such that every $\cc_k$ is supported on a subset of the support of $\DDelta$, and for every edge $e$, if $\cc_k(e)$ is non-zero, the sign of $\cc_k(e)$ agrees with that of $\DDelta(e)$.
\end{fact}

\paragraph{Dynamic Graphs and Recourse.} Dynamic graphs are graphs that undergo updates, i.e., change their state over time. 

\begin{definition}[Graph Updates]
    A graph update maps a graph $G$ to a new graph $G'$. We consider the following types of updates. 
    \begin{itemize}
        \item edge deletion: $V(G') = V(G)$ and $\exists e \in E(G): E(G') = E(G) \setminus \{e\}$
        \item edge insertion: $V(G') = V(G)$ and $\exists e \in E(G'): E(G) = E(G') \setminus \{e\}$ (we allow self-loops)
        \item vertex split: $\exists u \in V(G), \exists v,w \in V(G'): V(G) \setminus \{u\} = V(G') \setminus \{v,w\}$, $(v,w) \not \in E(G)$ and the graph obtained from contracting $\{v, w\}$ in $G'$ is $G$.
    \end{itemize}
\end{definition}

\begin{definition}[Dynamic Graph]
    We represent dynamic graphs as a sequence of graphs alongside a sequence of update batches $G^{(T)} \defeq \left(\left(G(t)\right)_{t \in [T]}, \left(B(t)\right)_{t \in [T - 1]}\right)$. Every update batch $B(t)$ contains a (possibly empty) sequence of graph updates such that applying them in-order to $G(t)$ yields $G(t + 1)$. 
\end{definition}

We will often keep the update sequence implicit for convenience. 

\begin{definition}[Restricted Dynamic Graphs]
    We call a dynamic graph $G^{(T)}$ decremental if no update batch contains an edge insertions and we call a dynamic graph $G^{(T)}$ incremental if all updates are edge insertions. Finally, we call a graph edge-dynamic if there are no vertex splits. 
\end{definition}

Whenever we write statements about two dynamic graphs with the same time horizon, they hold for all times. For example, $H^{(T)} \subseteq G^{(T)}$ should be interpreted as $H(i) \subseteq G(i)$ for all $i \in [T]$. We sometimes drop the time horizon $T$ for convenience, so $H \subseteq G$ should be interpreted as $H(i) \subseteq G(i)$ for all $i$ in the implicit time horizon $T$. 

\begin{definition}
    For a dynamic graph $G^{(T)}$, we let the recourse be the total number of updates in its update batches. We use the symbol $\rec{G^{(T)}}$ to refer to the recourse.
\end{definition}

\section{Min-Ratio Cycles via Distance Oracles}

The goal of this section is to describe how we can use a distance oracle to find approximate min-ratio cycles in the dynamic graph. Such min-ratio cycles have been used in \cite{maxflow, incrflow} to implement the iterations of the $\ell_1$-IPM. The purpose of this section is to describe how to find them statically in the distance oracle introduced in \cite{KMPG24}.\footnote{The distance oracles in \cite{KMPG24} build on the distance oracles given in \cite{ASZ} which a reminiscent of the classic Thorup-Zwick Distance Oracles \cite{thorup2005approximate}.} 

Our proof relies solely on the structural properties of the distance oracle. Since \cite{KMPG24} shows how to maintain these structural properties efficiently in dynamic graphs, it is then straightforward to dynamize our algorithm.

\subsection{The Vertex Sparsifier}

We next outline the vertex sparsification procedure from \cite{ASZ}, which forms the first building block of the distance oracles in \cite{KMPG24}. We point out that the construction has dependencies in the maximum degree of the input graph $G$; standard techniques can control this, as we discuss later. The idea is to map short path segments that live in local balls onto a terminal set $A \subseteq V$. 

\begin{definition}[Pivot Function and Local Balls]
For a graph $G = (V, E, \ll, \gg)$ and set $A \subseteq V$, we denote by $p_A : V \mapsto A$ the pivot function that maps each vertex $v \in V$ to the closest vertex $w \in A$, i.e. $\dist_G(v, p(v)) \leq \min_{w \in A}\dist_G(v,w)$. Ties are broken arbitrarily but consistently. We further define the local ball ${B}(v, A) = \{ w \in V\;|\; \dist_G(v,w) < \dist_G(v, p_A(v))\}$ and the local ball neighborhood $B^{\mathcal{N}}(v,A) = {B}(v, A) \cup \mathcal{N}({B}(v, A))$, i.e. the ball around $v$ and all incident vertices. We define the cluster $C(v, A) \defeq \{u \in V: dist_G(u,v) < dist_G(u,p_A(u)) \}$. We sometimes omit the subscript when the pivot function is clear from the context.
\end{definition}

The key idea is then to take shortest paths $\pi_G(v,u)$ for $u \in B(v, A)$, extend them by an edge $(u,w) \in E$ (not necessarily in the ball), and then 'project' them onto their pivots by prepending the shortest path $\pi_G(p_A(v), v)$ and appending the path $\pi_G(w, p_A(w))$. The resulting path $P$ from $p(v)$ to $p(w)$ is then converted into an edge $e = (p(v), p(w))$ and added to the vertex sparsifier over the set $A$ with length $\ll'(e) = \ll(P)$. The formal construction is given below.

\begin{definition}[ASZ-Paths and Vertex Sparsifier] \label{def:asz_path}
Given $G= (V, E, \ell, \gg)$, the generating path set is given by
\begin{align*}
    \mathcal{P} = \{ \pi_G(v,u) \oplus (u,w)  \;|\; v \in V, u \in B_G(v, A), w \in \mathcal{N}(u) \}.
\end{align*}
For an $xy$-path $P$ and set $A$, we denote by $\pproj_A(P)$, the path $\pi_G(p(x), x) \oplus P \oplus \pi_G(y, p(y))$. We define the set of ASZ paths by $\tilde{\mathcal{P}} = \{ \pproj_A(P) \;|\; P \in \mathcal{P}\}$. 

We define the ASZ vertex sparsifier by $H = (A, E', \ell')$ where there is a one-to-one correspondence between $xy$-paths $P \in \tilde{\mathcal{P}}$ and edges $e = (x,y) \in E'$ of length $\ell'(e) = \ell(P)$ and the gradient as the signed gradient of the path, i.e. $\gg'(e) = \gg^{\top} \vecone_P$ where $\vecone_P$ is $\pm 1$ in coordinate $e' \in E$ if $e' \in P$ and otherwise $0$.
\end{definition}
\begin{remark}
Note that $G$ and the set $A$ uniquely determine the generating path set $\mathcal{P}$ and the ASZ paths $\tilde{\mathcal{P}}$ and thus the ASZ vertex sparsifier $H$.
\end{remark}

We fix $G$ and $A$ for the rest of the section and denote by $H$ the ASZ vertex sparsifier. We next prove that for any $u,v \in V$, either they are locally close, i.e. $u \in B(v, A)$, or $H$ preserves distances between them. 

\begin{claim}\label{clm:ASZstretch}[see \cite{ASZ}]
For any $x \in V$, $y \not\in B(x, A)$, $$\dist_G(x,y) = \Theta(\dist_G(x, p(x)) + \dist_H(p(x), p(y)) + \dist_G(p(y), y).$$
\end{claim}
\begin{proof}
Let us first prove the claim for the special case where $x,y \in A$ and thus $\dist_G(x, p(x)) = \dist_G(y, p(y)) = 0$.  Here, for a shortest $xy$-path $\pi_G(x,y)$ in $G$, we define the subsequence $v_0 = x, v_1, v_2, \ldots, v_\tau = y$ of $\pi_G(x,y)$ where $v_{i+1}$ is the vertex closest to $y$ on the path $\pi_G(x,y)$ such that $v_i v_{i+1}$-path segment $P_i = \pi_G(x,y)[v_i, v_{i+1}]$ is in $\mathcal{P}$. 

We claim that for $i < \tau$, $\dist_G(v_i, p(v_i)) \leq \dist_G(v_i, v_{i+1})$. From $v_i$, let $u_i$ be the last vertex on $\pi_G(x,y)$ after $v_i$ that is in $B(v_i, A)$. If $u_i \neq y$, let $w_i$ be the vertex on $\pi_G(x,y)$ following $u_i$. By construction, we have that the path $P_i' = \pi_G(v_i, u_i) \oplus (u_i, w_i) \in \mathcal{P}$ and thus $P'_i \subseteq P_i$. But since $w_i$ is outside of the ball $B(v_i, A)$, we have $\dist_G(v_i, w_i) \geq \dist_G(v_i, A) = \dist_G(v_i, p_A(v_i))$. If $u_i = y$, we also have $v_{i+1} = u_i = y$. Then, trivially from $y \in A$, we have $\dist_G(v_i, p(v_i)) = \dist_G(v_i, A) \leq \dist_G(v_i, v_{i+1})$. 

By the triangle inequality,  $\dist_G(v_{i+1}, p_A(v_{i+1})) \leq \dist_G(p_A(v_i), v_i) + \dist_G(v_i, v_{i+1})$. Combining this with the subpath property of shortest paths and our insight above yields
\begin{align}
    \ell(\pproj_A(P_i)) &\leq \dist_G(p_A(v_i), v_i) + \dist_G(v_i, v_{i+1}) + \dist_G(v_{i+1}, p_A(v_{i+1})) \nonumber \\&\leq 2(\dist_G(p_A(v_i), v_i) + \dist_G(v_i, v_{i+1})) \nonumber \\& \leq 4 \cdot \dist_G(v_i, v_{i+1}). \label{eq:keyInequalityPathProj}
\end{align}
This, in turn, yields 
\[
    \dist_H(x,y) \leq \sum_{i = 0, \ldots, \tau-1} \ell'(\pproj_A(P_i)) \leq \sum_{i = 0, \ldots, \tau-1}  4\dist_G(v_i, v_{i+1}) = 4 \cdot \dist_G(x,y). 
\]
A lower bound $\dist_G(x,y) \leq \dist_H(x,y)$ is immediate from the fact that we project $uv$-paths to $(u,v)$ edges of equal length.

Finally, consider the more general case, where $x \in V, y \not\in B(x,A)$. Then, since $y \not\in B(x, A)$, $\dist_G(x,y) \geq \dist_G(x, p(x))$ and thus by triangle inequality, $\dist_G(y, p(y)) \leq \dist_G(x,y) + \dist_G(x, p(x))$. Since the distance between terminals $p(x)$ and $p(y)$ is approximately preserved, this yields 
\[
\dist_G(x, p(x)) + \dist_H(p(x), p(y)) + \dist_G(p(y), y) = \Theta(\dist_G(x,y)).
\]
\end{proof}

Finally, to avoid that many vertex pairs $u,v \in V$ are locally close and thus require explicit storage, we want to choose the terminal set $A$ such that each local ball $B(v, A)$ is of small size. A natural idea is to use a standard hitting set argument: by adding each vertex $w \in V$ to $A$ with some probability $p = \Theta(k/n)$, we have with high probability that at least one of the $k$-closest vertices of $v$ is 'hit', i.e. added to $A$. This enforces $|B(v,A)|< k$ while $A$ has expected size $\tilde{O}(n/k)$. In \cite{roditty2005deterministic}, it was shown that such a hitting set can also be constructed deterministically. Their main result is stated below.

\begin{theorem}[see \cite{roditty2005deterministic}]
Given a graph $G=(V, E, \ell)$ and integer $k \geq 1$, we say $A$ is $k$-shattering if for all $v \in V$, $|B(v, A)|, |C(v,A)| \leq k$, i.e. if the $(k+1)$-th closest vertex to $v$ has distance greater-or-equal to its pivot $p_A(v)$. A $k$-shattering set $A \subseteq V$ of size $O(n \log n/k)$ can be computed in time $\tilde{O}(mk)$ along with all local balls $B(v, A)$ for $v \in V$ and distances $\dist_G(u,v)$ for $u \in B(v,A)$.
\end{theorem}

This yields that we can store the distances of all locally close vertex pairs $u,v \in V$ in time and space $\tilde{O}(m k)$. However, choosing $A$ to be $k$-shattering also has a second crucial effect: it allows us to bound the number of edges in $H = (A, E', \ell')$ by $\tilde{O}(m \cdot \poly(k))$. Here, we assume w.l.o.g. that $G$ has constant maximum degree $\Delta$ which yields 
\[
    |E'| \leq |\mathcal{P}| \leq  \Delta \cdot \sum_{v \in V} |B(v,A)| \leq \Delta \cdot n k.
\]

The \cite{KMPG24}-vertex sparsifier combines the ASZ paths with core graphs (for the reason of efficient dynamic maintenance, for a static built, ASZ paths on their own would already suffice). Here a core graph is a smaller graph that approximates the distance structure. We first recall the definition of rooted forests from the preliminaries. 

\begin{definition}[Rooted Forest] \label{def:rooted_forest}
Given a graph $G = (V, E)$, a rooted forest $F \subset G$ is an acyclic subgraph of $G$ with additional root function $\root_F : V \mapsto V$ that maps all vertices in the same tree in $F$ to a unique vertex in the tree.
\end{definition}

Next, we define how to project an edge onto a forest. If both endpoints are in the same tree, its projection is the cycle it forms with the tree. Otherwise, its projection is given by concatenating the edge with the paths to the respective tree roots. 

\begin{definition}[Forest Projection]\label{def:forestProj}
        Given $G = (V, E)$ and a rooted forest $F$ of $G$, and an edge $e = (u, v)$ such that $u \in T_i$ and $v \in T_j$ where $T_i$ and $T_j$ denote the $i$-th and $j$-th tree in $F$ respectively, we let 
        \begin{align*}
            \pproj_F(e) \defeq T_i[\root_F(u), u] \concat (u,v) \concat T_j[v, \root_F(v)]. 
         \end{align*}
         If $i = j$, we call $e$ a tree edge and $\pproj_F(e)$ a tree cycle.
\end{definition}

Stretch is an important notion that governs the way we pick our forests. We typically aim to minimize stretch as much as possible. It is defined as follows.  

\begin{definition}[Stretch] \label{def:stretch}
    Given $G = (V, E, \ll)$ and a rooted forest $F$
    \begin{align*}
        \str_F(e = (u, v)) \defeq \|\LL\vecone_{\pproj_F(e)}\|_1/\ll(e). 
    \end{align*}
    Then, we define the stretch of a path as follows. 
    \begin{align*}
        \str_F(P) \defeq \frac{1}{\ll(P)} \sum_{e \in P} \str_F(e) \cdot \ll(e)
    \end{align*}
Whenever we refer to upper bounds for the stretch instead of the true stretch, we use $\widetilde{\str}(\cdot)$. 
\end{definition}

Given a graph $G$ and a forest $F \subseteq G$, a core graph is obtained by contracting the forest components. A core graph sufficiently reduces the number of vertices, but doesn't preserve all distances sufficiently well. We show that one core graph preserves roughly half of the not yet preserved ASZ-paths. Therefore the full vertex-sparsifier will be obtained by combining logarithmically many core graphs. 

\begin{definition}[Core Graph]
    Given a rooted forest $F \subseteq G$ for some $G$, we let $\mathcal{C}(G, F)$ be the graph obtained by contracting every tree component $F$ in $G$. Any subgraph of $G$ is projected to a subgraph of $\mathcal{C}(G, F)$ by contraction of the forest F, we denote this projection map by $\Pi_{\mathcal{C}}: 2^{E(G)} \rightarrow 2^{E(\mathcal{C}(G, F))}$.
\end{definition}

\begin{remark}
    The map $\Pi_{\mathcal{C}}$ induces a 1:1 correspondence between the off-tree edges and the edges in $\mathcal{C}(G,F)$. For any off-tree edge $e$ we understand $\Pi_{\mathcal{C}}(e)$ as an edge (potentially a self-loop) in $\mathcal{C}(G,F)$.
\end{remark}

In this section, we will be interested in core graphs with low stretch. Such a core graph comes with stretch upper bounds $\tilde{\str}(\cdot)$, and guarantees that the stretch of every edge $e$ is bounded by $\tilde{\str}(\cdot)$. We consider stretch upper bounds instead of directly bounding the stretch because these will be much more stable when considering dynamic graphs. 

\begin{definition}[$\ell_1$-Core Graph]\label{def:L1CoreGraph}
    Given a rooted forest $F \subseteq G$ with stretch upper bounds $\tilde{\str}(\cdot)$ for some $G = (V, E, \ll, \gg)$, we call the core graph $\mathcal{C}(G, F)$ an $\ell_1$-Core Graph if for every off-tree edge $e \in E(G)$ we have $\sum_{e' \in \pproj_F(e)} \ll(e') \leq \ll_{\mathcal{C}}(\Pi_{\mathcal{C}}(e)) \defeq \tilde{\str}(e) \cdot \ll(e)$ and $\gg_{\mathcal{C}}(\Pi_{\mathcal{C}}(e)) \defeq \sum_{e' \in \pproj_F(e)} \gg(e')$.
\end{definition}

The next lemma shows that a few forests suffice to preserve all the paths in any concrete path set. Since we are interested in approximately preserving distances, we will always instantiate the lemma with the ASZ-paths $\tilde{\mathcal{P}}$ defined above, but we emphasize that we do not use any specifics about the path set.

\begin{lemma}[Core Graphs, Theorem 3.1 of \cite{KMPG24}]\label{lm:StaticCoreGraphs}
    Given a graph $G$, where $\Delta$ bounds the maximum degree and some parameter $k$, there exists an $\beta = \tilde{O}(1) \cdot \Delta \cdot k$ shattering pivot set $A$, $\eta = \tilde{O}(1)$ forests $F_1, \ldots F_{\eta}$ of $G$ and core graphs $\mathcal{C}(G,F_1), \dots, \mathcal{C}(G,F_\eta)$, together with stretch bounds $\widetilde{\str}_i(e)$ for all $e \in E$ such that 
    \begin{enumerate}
        \item for all $i,$ each tree of $F_i$ has at most $\tilde{O}(k)$ many vertices and the set of roots of $F_i$ has size at most $\tilde{O}(n/k)$, and contains all vertices in $A$,
        \item every ASZ-path $P \in \tilde{\mathcal{P}}$ (see \Cref{def:asz_path}) is preserved, i.e. $\forall P \in \mathcal{P}: \exists i: \widetilde{\str}_i(P) \leq \alpha = \tilde{O}(1).$ 
    \end{enumerate}
    The algorithm runs in time $\tilde{O}(\beta^{O(1)} \cdot n)$.
\end{lemma}

In \cite{KMPG24}, the construction is dynamized, which allows us to maintain the core graphs in \Cref{sec:Dynamization}.

\paragraph{Adding together the core graphs. }

We next define the core graph sum. It takes multiple core graphs (that together preserve all ASZ-paths) and glues them together at the roots with a clique of length and gradient zero. This leads to our final vertex sparsifier.

\begin{definition}[Core Graph Sum]\label{def:StaticCoreGraphSum}
    Given a graph $G$, the path set $\mathcal{P}$, and the set of $\eta$ forests $F_1, \ldots, F_{\eta}$ of $G$ from Lemma \ref{lm:StaticCoreGraphs}, we denote by $\mathcal{C}(F_1, \ldots, F_{\eta})$ the graph obtained from $\mathcal{C}(G, F_1)(t) + \ldots + \mathcal{C}(G, F_\eta)(t)$ by connecting all $\eta$ copies of a vertex $a \in A$ by a star. The edges of the star have length $0$ and gradient $0$.
\end{definition}

\begin{remark}\label{rmk:StarEmbeddingExtension}
    We understand $\mathcal{C}(G, F_i)(t)$ as a subgraph of $\mathcal{C}(F_1, \ldots, F_{\lambda})(t)$. In that sense, $\Pi_{\mathcal{C}_i}$ maps into $\mathcal{C}(F_1, \ldots, F_{\lambda})$. Here, we always extend the mapping to the center of the star, i.e. assume that the mapping of some path $P$ is given by $\Pi_{\mathcal{C}_i}(P) = (v_0,v_1, \dots, v_k)$ into $\mathcal{C}(G, F_i)(t)$ then we extend it into $\mathcal{C}(F_1, \ldots, F_{\lambda})(t)$ as $\Pi_{\mathcal{C}_i}(P) = (s_0, v_0,v_1, \dots, v_k, s_k)$ where $s_0,s_k$ denotes the star centers associated with $v_0, v_k$.
\end{remark}

\subsection{The Vertex Sparsification Dichtomy} \label{sec:cycle_preservation}

In this section, we prove that instantiating \Cref{lm:StaticCoreGraphs} with the ASZ-paths from \Cref{def:asz_path} allows us to either detect an approximate min-ratio cycle in some local ball neighborhood $B^{\mathcal{N}}(v,A)$ or as a tree cycle in one of the forests $F_1, \dots, F_{\eta}$, or guarantee that the core graph sum of $\mathcal{C}(G,F_1), \dots, \mathcal{C}(G,F_\eta)$ contains an approximate min-ratio cycle. This is the sub-problem we are required to solve repeatedly to obtain a primal min-cost flow algorithm via the $\ell_1$-IPM framework. 

We start by defining a mapping of the cycles in $G$ into $H$. The definition extends the map for shortest paths from $G$ into $H$ that was previously introduced in the proof of \Cref{clm:ASZstretch}.

\begin{definition}[Pivot Cycle]\label{def:pivotCycle}
Given a cycle $C = \langle x = x_1, x_2, \ldots, x_\eta = y\rangle$ in $G$. We define the subsequence $v_0 = x, v_1, v_2, \ldots, v_\tau = y$ of $C \setminus (y,x)$ where $v_{i+1}$ is the vertex closest to $y$ on the path $C \setminus (y,x)$ such that $v_i v_{i+1}$-path segment $P_i = \pi_G(x,y)[v_i, v_{i+1}]$ is in $\mathcal{P}$ and define $P_{\tau} = (y,x)$. We let $\pproj_A(C) \defeq \pproj_A(P_0) \concat \pproj_A(P_1) \concat \cdots \concat \pproj_A(P_{\tau})$ denote its pivot cycle. 
\end{definition}

We note that edges are trivially in $\mathcal{P}$ and thus in particular, all segments $P_i$ above are in $\mathcal{P}$. Since we merely extend the definition, it is not hard to show the following claim.

\begin{lemma}\label{clm:projectedCycleLength}
For cycle $C$ and vertex $x$ on the cycle, if $C \not\subseteq B(x,A)$,then $\sum_{e \in \pproj_A(C)} \ell(e) \leq 8 \cdot \ell(C)$.
\end{lemma}
\begin{proof}
The inequality \eqref{eq:keyInequalityPathProj} in \Cref{clm:ASZstretch} extends to segments $P_i$ for $i < \tau$. Thus, by the same argument $\sum_{i < \tau} \ell(\pproj(P_i)) \leq 4 \cdot \ell(C)$. Further, we have that the segment $P_{\tau}$ has $\pproj(P_{\tau}) = \pi_G(p(y), y) \concat (y,x) \concat \pi_G(x, p(x))$. By assumption, at least one vertex $x'$ on $C$ is not in $B_G(x,A)$ and thus $\ell(C) \geq \dist_G(x, x') \geq \dist_G(x, p_A(x))$. We thus get by the triangle inequality
\begin{align*}
    \ell(\pproj(P_{\tau})) &= \dist(p_A(y), y) + \ell(y,x) + \dist_G(x, p_A(x))\\
    &\leq 2(\ell(y,x) + \dist_G(x, p_A(x)))\\
    &\leq 4 \cdot \ell(C).    
\end{align*}
\end{proof}

It is not hard to observe that the set of pivot cycles is preserved in $H$ and can be mapped back along the underlying path segments into $G$ (this might make the cycle non-simple; however, this is not crucial for our purpose). 

We next prove the dichotomy that either an approximate min-ratio cycle is contained in a local ball or can be recovered in the vertex sparsifier $H$.

\begin{claim}[ASZ Dichotomy]\label{thm:ASZ-Dichotomy1}
Let $C$ be any cycle in $G = (V,E,\ll,\gg)$. Let $C'$ be the best min-ratio cycle fully contained in a local ball neighborhood $B^{\mathcal{N}}(v,A)$ for some $v \in V$. Then,
    \begin{align*}
        \min\left( \frac{\gg^T \vecone_{\pproj_A(C)}}{\norm{ \LL \vecone_{\pproj_A(C)}}_1}, \frac{\gg^T \vecone_{C'}}{\norm{ \LL \vecone_{C'}}_1}\right)
        \leq \frac{1}{10} \cdot \frac{\gg^T \vecone_C}{\norm{\LL \vecone_C}_1}.
    \end{align*}
In particular, the claim holds for the optimal min-ratio cycle $C^*$ in $G$ and thus, either an approximate min-ratio cycle is in a local ball, or we can recover one from $H$.
\end{claim}
\begin{proof}
We write  $C = \langle x = x_1, x_2, \ldots, x_\eta = y\rangle$ and let $v_0 = x, v_1, \ldots, v_\tau = y$ be defined as in \Cref{def:pivotCycle}.

Let $C[v_i, v_{i+1}]$ be the segment of $C$ from $v_i$ to $v_{i+1}$. Let $\pi_G(v_{i+1}, v_i)$ be the shortest path between these vertices in $G$. We define the circulation $C_i = C[v_i, v_{i+1}] \cup \pi_G(v_{i+1}, v_i)$ is contained in $B_G^{\mathcal{N}}(v_i,A)$ since $v_{i+1}$ is the first (and only) vertex in $C_i[v_i, v_{i+1}]$ that is not in $B_G(v_i, A)$ by definition. Since $C \subseteq G$, we have $\|\LL \vecone_{C_i}\| \leq 2 \cdot \|\LL \vecone_{C[v_i, v_{i+1}]}\|$.

Note that since we define $\pproj_A(C) = \pproj_A(P_0) \concat \pproj_A(P_1) \concat \cdots \concat \pproj_A(P_{\tau-1}) \concat \pproj_A((y,x))$ and the end vertex of $P_i$ is the starting vertex of $P_{i+1}$ (analogous for the wrap around), we have that for each such endpoint $v_i$, the pivot paths $\pi_G(v_i, p_A(v_i))$ is traversed in the forward direction on $P_i$ and in the backward direction on $P_{i+1 \mod \tau +1}$. Thus, $\vecone_{\pproj_A(C)} = \sum_{i = 0, \ldots, \tau} \vecone_{\pi_G(v_i, v_{i+1} \mod \tau + 1)}$. Since $C_i = C[v_i, v_{i+1}] \cup \pi_G(v_{i+1}, v_i)$, we recover the original cycle $C$ by subtracting $\vecone_{\pproj_A(C)} - \sum_i \vecone_{C_i} = \vecone_{C}$. 

Next, note that by \Cref{clm:projectedCycleLength}, if $C$ is not contained in the ball of some vertex on the cycle, we have that $\sum_{e \in \pproj_A(C)} \ell(e) \leq 8 \cdot \|\LL \vecone_C\|_1$. Together with $\|\LL \vecone_{C_i}\| \leq 2 \cdot \|\LL \vecone_{C[v_i, v_{i+1}]}\|$, we obtain $\sum_{e \in \pproj_A(C)} \ll(e) + \sum_{i} \sum_{e \in C_i} \ll(e) \leq 10 \norm{\LL \vecone_C}_1$. We obtain
\begin{align*}
    \frac{\vecone_{\pproj_A(C)}^\top \gg - \sum_{i} \vecone_{C_i}^\top \gg}{\sum_i \sum_{e \in \pproj_A(C)} \ll(e) + \sum_{i} \sum_{e \in C_i} \ll(e) } \leq \frac{1}{10} \frac{\vecone_C^\top \gg}{\norm{\LL \vecone_C}_1}.
\end{align*}
The claim now follows from \Cref{fact:ratio_inequality}. Note that here we ignored a technical detail: $C_i$ might not be a cycle but rather a circulation, but in this case, we can further decompose to obtain the claim.
\end{proof}

\begin{theorem}[Forest Dichotomy]\label{thm:ASZ-Dichotomy2}
    Given a dynamic graph $G = (V,E,\ll,\gg)$, and pivot cycle $C$ in $G(t)$, then some tree cycle $C_0$ of one of the forests $F_1(t), \dots, F_{\eta}(t)$ or some circulation $C_1$ in the core graph sum $\mathcal{C}(F_1, \dots, F_\eta)(t)$ has comparable ratio, i.e.
    \begin{align*}
        \min\left( \frac{\gg^T \vecone_{C_0}}{\| \LL \vecone_{C_0}\|_1}, \frac{\gg_{\mathcal{C}}^T \vecone_{C_1}}{\|\LL_{\mathcal{C}} \vecone_{C_1}\|_1}\right)
        \leq \frac{1}{2 \cdot \alpha} \cdot \frac{\gg^T \vecone_{C}}{\norm{\LL \vecone_{C}}_1}.
    \end{align*}
\end{theorem}

\begin{proof}
    We observe that for the pivot cycle $C = \pproj_A(P_0) \concat \pproj_A(P_1) \concat \cdots \concat \pproj_A(P_{\tau-1}) \concat \pproj_A(P_{\tau})$, where all path segments $\pproj(u_i, u_{i+1}) \in \tilde{\mathcal{P}}$ by \Cref{def:asz_path} and \Cref{def:pivotCycle}. To ease notation, we write $\tilde{P}_i = \pproj_A(P_i)$. By \Cref{lm:StaticCoreGraphs}, item 4, there exists a forest $F_{\sigma(i)}$ with its associated core graph $\mathcal{C}_{\sigma(i)}$ such that $\tilde{P}_i$ is preserved, i.e. $\widetilde{\str}_{\sigma(i)}\left(\tilde{P}_i\right) \leq \alpha$. As a first step, we decompose the cycle $C$ into tree-edges $C_T$ and off-tree edges, i.e. $C_T = \bigcup_i \{e = (u,v) \in \tilde{P}_i \mid \exists T_j \subseteq F_{\sigma(i)}: u,v \in T_j \}$. If we denote for all $e = (u,v) \in C_T \cap \tilde{P}_i$ by $C_e$ the tree cycle $\pproj_{F_{\sigma(i)}}(e)$, then 
    \begin{align*}
        \vecone_{C} &= \sum_{e \in C_T} \vecone_{C_e} + \sum_i \sum_{e \in \tilde{P}_i \setminus C_T} \vecone_{\pproj_{F_{\sigma(i)}}(e)} \\
        &= \sum_{e \in C_T} \vecone_{C_e} + \vecone_{\tilde{C}},
    \end{align*}
    where $\tilde{C} = \concat_i \concat_{e \in \tilde{P}_i\setminus C_T} \pproj_{F_{\sigma(i)}}(e)$.\\
    \\
    Moreover, we note that 
    \begin{align*}
        \|\LL_{\mathcal{C}} \vecone_{C_1}\|_1 + \sum_{e \in C_T} \|\LL \vecone_{C_e}\|_1 &\leq \sum_{i} \left( \|\LL_{\mathcal{C}} \vecone_{\Pi_{\mathcal{C}_{\sigma(i)}}(\tilde{P}_i)}\|_1 + \sum_{e \in C_T \cap \tilde{P}_i} \ll(e) + \ll(\pproj_F(e))\right) \\
        &\leq \sum_{i} \left( \sum_{e \in C \setminus C_T \cap \tilde{P}_i} \tilde{\str}_i(e) \cdot \ll(e) + \sum_{e \in C_T \cap \tilde{P}_i} \left(1 + \tilde{\str}_i(e)\right) \cdot \ll(e) \right)\\
        &\leq 2 \cdot \alpha \cdot \|\LL \vecone_{C}\|_1 ,
    \end{align*}
    where $C_1 = \Pi_{\mathcal{C}_{\sigma(1)}}(\tilde{P}_1) \concat \dots \concat \Pi_{\mathcal{C}_{\sigma(\eta)}}(\tilde{P}_\eta)$ \footnote{Here, we point out that each embedding $\Pi_{\mathcal{C}_{\sigma(i)}}(\tilde{P}_i)$ is extended to the star center according to Remark \ref{rmk:StarEmbeddingExtension}.} in $\mathcal{C}(F_1, \dots, F_\eta)$.
    Putting these two facts together and using \Cref{fact:ratio_inequality}, we obtain
    \begin{align*}
        \min \left\{\frac{\gg^T \vecone_{\tilde{C}}}{\|\LL_{\mathcal{C}} \vecone_{C_1}\|_1}\right\} \cup \left\{\frac{\gg^T \vecone_{C_e}}{\|\LL \vecone_{C_e}\|_1} : e \in C_T\right\} &\leq \frac{\gg^T\vecone_{\tilde{C}} + \sum_{e \in C_T} \gg^T\vecone_{C_e} }{\|\LL \vecone_{C_1}\|_1 + \sum_{e \in C_T} \|\LL \vecone_{C_e}\|_1} \\
        &\leq \frac{1}{2 \alpha} \frac{\gg^T\vecone_{C^p} }{\|\LL \vecone_{C^p}\|_1}.
    \end{align*}
    To conclude the argument, we observe that the projected cycle $C_1$ has the same gradient as the cycle $\tilde{C}$, i.e. $\gg^T \vecone_{\tilde{C}} = \sum_i \sum_{e \in \tilde{P}_i \setminus C_T} \gg^T\vecone_{\pproj_{F_{\sigma(i)}}(e)} = \sum_i \sum_{e \in \tilde{P}_i \setminus C_T} \gg_{\mathcal{C}}^T\vecone_{\Pi_{\mathcal{C}_{\sigma(i)}}(e)} = \gg_{\mathcal{C}}^T \vecone_{C_1}$.
\end{proof}

Stacking the two ASZ Dichotomies Theorem \ref{thm:ASZ-Dichotomy1} and Theorem \ref{thm:ASZ-Dichotomy2} on top of each other, we obtain that two things can happen: either there is a good cycle that can be found efficiently, either in some local ball neighborhood $B^{\mathcal{N}}(v,A)$ or as a tree cycle in one of the forests $F_1, \dots, F_{\eta}$, or we preserved an approximate min-ratio cycle in the next (significantly smaller) level of the core graph hierarchy. 

\begin{corollary}[Vertex Sparsifier Dichotomy]\label{Cor:ASZ-Trichotomy}
    Let $C$ be any cycle in $G = (V,E,\ll,\gg)$. 
    Then there exist $C_0$ fully contained in a local ball neighborhood $B^{\mathcal{N}}(v,A)$ for some $v \in V$, a tree cycle $C_1$ of one of the forests $F_1(t), \dots, F_{\eta}(t)$ and a circulation $C_2$ in the core graph sum $\mathcal{C}(F_1, \dots, F_\eta)(t)$ such that 
    \begin{align*}
        \min\left( \frac{\gg^T \vecone_{C_0}}{\| \LL \vecone_{C_0}\|_1}, \frac{\gg^T \vecone_{C_1}}{\| \LL \vecone_{C_1}\|_1}, \frac{\gg_{\mathcal{C}}^T \vecone_{C_2}}{\|\LL_{\mathcal{C}} \vecone_{C_2}\|_1} \right)
        \leq \frac{1}{20 \alpha} \cdot \frac{\gg^T \vecone_C}{\norm{\LL \vecone_C}_1}.
    \end{align*}
\end{corollary}

\paragraph{The Edge Sparsification Dichotomy.} After we have established vertex sparsification, we next need to reduce the edge count to obtain a clean recursion. An edge sparsifier $H \subseteq G$ that approximately preserves distances is called a spanner. 

\begin{definition}[Spanner]
    A connected subgraph $H$ of a graph $G = (V, E, \ll)$ with graph embedding $\Pi_{G \rightarrow H}$ is a $\delta$-spanner if 
    \begin{align*}
        \ll\left(\Pi_{G \rightarrow H}((u, v))\right) \leq \delta \dist_G(u,v).
    \end{align*}
\end{definition}

We remark that the subgraph $H$ inherits the lengths and gradients from $G$. 

\begin{theorem}[The Edge Sparsifier Dichotomy]\label{thm:SpannerDichotomy}
    Consider a $\delta$-approximate spanner $H$ of $G$. For every edge in $e \in G \setminus H$, consider the cycle formed by $e = (u,v)$ and $\Pi_{G \rightarrow H}(e)$. Then, either
    \begin{itemize}
        \item one of these cycle is a $1/2\delta$ approximate min ratio cycle, or
        \item there is a $1/2\delta$-approximate min-ratio cycle in $H$.
    \end{itemize}
\end{theorem}
\begin{proof}
    Consider the min ratio cycle $C$ in $G$. We lift the cycle graph to $H$ by replacing every edge $e = (u,v)$ that is not present in $H$ with the path $\Pi_{G \rightarrow H}(e)$. For each such edge, we also consider the cycle $C_e$ formed by $e$ and $\pi_G(u, v)$, but traversed in the opposite direction. Then
    \begin{align*}
        \vecone_C = \sum_{e \in C} \vecone_{C_e} - \vecone_{C'}.
    \end{align*}
    Since $\ll(C) + \sum_e(\ll(C_e)) \leq 2\delta \ll(C)$, we obtain the result by \Cref{fact:ratio_inequality}. 
\end{proof}

Therefore, we can reduce the number of edges with a spanner, and it suffices to check all edges that get sparsified away. To ensure that the number of edges gets reduced appropriately, we chose $\delta = O(\log n)$.

\begin{theorem}\label{thm:StaticSpanner}
Given an $m$-edge $n$-vertex unweighted, undirected graph $G$ consisting only of edge
deletions and $\widetilde{O}(n)$ vertex splits. There is a deterministic
algorithm with parameter $1 \le L \le o\left(\log^{1/6} m / \log\log m \right),$ that computes a $\gamma^{O(L)}$-spanner $H$ with at most $n\gamma$ edges such that for
some $\gamma = \exp(O(\log^{2/3} m \cdot \log\log m))$

\[
    \operatorname{card}(\Pi_{G \to H}) \le \gamma^{O(L)} \quad \text{ and } \quad \vcong(\Pi_{G \to H})
    \le
    \gamma^{O(L^2)} \Delta.
\]
The algorithm takes time $\widetilde{O}(m \gamma)$.
\end{theorem}

\subsection{Min-Ratio Cycle from the Vertex Sparsifier Hierarchy}

Given both the Vertex Sparsifier Dichotomy (\Cref{Cor:ASZ-Trichotomy}) and the Edge Sparsifier Dichotomy (\Cref{thm:SpannerDichotomy}) at hand it is straightforward to construct a hierarchy to identify an approximate min-ratio cycle.

We construct the same core graph hierarchy as in \cite{KMPG24} to locate an approximate min-ratio cycle. Since \cite{KMPG24} already shows how to maintain this hierarchy in dynamic graphs, it is then easy to augment their data structure to maintain approximate min-ratio cycles throughout the iterations of the $\ell_1$-IPM to solve incremental min-cost flow (See \Cref{sec:Dynamization}).

\paragraph{Construction of the Hierarchy.} We construct $\Lambda + 1$ hierarchy levels for
$\Lambda = \log^{1/21} m$, and at each level
$0 \le i \le \Lambda$, we construct two graphs $G_i, H_i$.
We let $G_0 = G$, and for $0 \le i \leq \Lambda$, $H_i$ is the spanner computed on $G_i$ according to Theorem \ref{thm:StaticSpanner} to reduce the edge count, while for $0 \leq i < \Lambda$, 
$G_{i+1}$ is the core graph sum with core graphs computed in Lemma \ref{lm:StaticCoreGraphs} on the graph
$H_i$ with size reduction parameter $k = m^{1/\Lambda}$.

\paragraph{Controlling Maximum Degrees.} So far, we have not properly discussed how to bound the maximum degrees of the graphs $G_i$. As pointed out previously, runtimes depend on the maximum degree of each $G_i$, so we have to control the degrees across the hierarchy. In static graphs, this turns out to be easy: we can split high-degree vertices and connect them with a $0$ length, $0$ gradient edge. We apply this process to $G_0$ initially, yielding a graph with max-degree $3$, which increases the vertex count to $\tilde{O}(m)$ while retaining the edge count at $\tilde{O}(m)$. We can further apply postprocessing to the spanner algorithm to split high-degree vertices. This yields a maximum degree $O(\gamma)$ for each $H_i$ while increasing the number of vertices and edges on each level by only by a constant factor. We note that in the dynamic maintenance of the above hierarchy, \cite{KMPG24} enforces the constraints with a similar trick; however, they force vertex splits by deleting edges from the core graphs. This requires a careful feedback loop between $G_i$ and $H_i$ to get eventual control over the degrees. Luckily, we can simply black box this technical nuisance here.

\paragraph{Augmenting the Hierarchy.} Given the Vertex Sparsifier Dichotomy \ref{Cor:ASZ-Trichotomy} and the Edge Sparsifier Dichotomy \ref{thm:SpannerDichotomy}, it is straightforward to construct a linked list of candidate cycles of almost linear size that contains an approximate min-ratio cycle. The only slightly tricky part is noticing how to check all cycles in a local ball neighborhood $B^{\mathcal{N}}(v,A)$. Given that $B^{\mathcal{N}}(v,A)$ can be of size $k \cdot \Delta$ there are too many cycles $2^{k \cdot \Delta}$ to include them all in our list. For this reason, we use the next Lemma to reduce the number of cycles we need to check in each $B^{\mathcal{N}}(v,A)$ to $\tilde{O}(k \cdot \Delta^2)$.

\begin{lemma}[Crude Approximation of Min-Ratio Cycle]\label{lem:brute_force}
    For a graph $G = (V, E, \ll, \gg)$ with $\ll \in \R_{\geq 1}^E$ and $\gg \in \R^E$, there is a list of $|E|$ simple cycles such that for any choice of gradients $\gg$ one cycle $\hat{\cc}$ in the list satisfies
    \begin{align*}
        \frac{\gg^\top \hat{\cc}}{\norm{\LL \hat{\cc}}} \leq \frac{1}{|V|} \min_{\cc: \BB_G \cc = \veczero} \frac{\gg^\top \cc}{\norm{\LL \cc}_1}
    \end{align*}
    in time $\tilde{O}(|E| + |V|)$
\end{lemma}
\begin{proof}
    Compute the minimum spanning tree of $G$. Then, decompose the optimal cycle $\cc^{\star}$ into tree cycles. The length of all tree cycles together is at most $|V|$ times the length of the original cycle, and their gradients sum up to the gradient of the original cycle. Therefore, the best tree cycle is a $1/|V|$ approximate min-ratio cycle. 
\end{proof}

This, in particular, implies that for any local ball neighborhood $B^{\mathcal{N}}(v,A)$ it suffices to check at most $\tilde{O}(k \cdot \Delta^2)$ many cycles. Finally, we formally analyze the cost of constructing the list of cycles that we need to check to find an approximate min-ratio cycle of $G$.

\begin{theorem}\label{thm:CandidateCyclesList}
    For reasonably large $m$, there exists an ordered linked list of candidate cycles of size $\tilde{O}(m \cdot k \cdot \gamma^2 )$ such that a cycle $C$ of these candidate cycles is an approximate min-ratio cycle, i.e.
    \[ \frac{\gg^\top \vecone_{C}}{\| \LL \vecone_C \|_1} \leq \frac{1}{\tilde{O}(k)}\left(\frac{1}{40 \alpha \cdot \delta}\right)^{\Lambda+1} \min_{\BB \cc = 0} \frac{\gg^\top \cc}{\|\LL \cc\|_1}. \]
    Each cycle is of length at most $\zeta = \tilde{O}(\max\{k, \gamma^{O(L)}\})$. The ordered linked list along with all ratios achieved by the contained cycles can be explicitly output in time $O(m \cdot k^2 \cdot \gamma^2)$.
\end{theorem}

\begin{proof}
    Given a level $i < \Lambda$ of the hierarchy, according to the Vertex Sparsifier Dichotomy \ref{Cor:ASZ-Trichotomy} and the Edge Sparsifier Dichotomy \ref{thm:SpannerDichotomy}, it suffices to check every tree cycle of the forests $F_{i,1}, \dots, F_{i,\eta_i}$, every simple cycle fully contained in a local ball neighborhood $B^{\mathcal{N}}_{H_{i-1}}(v,A)$ for some $v \in V(H_{i-1})$ and every cycle given by an edge $(u,v) \concat \Pi_{G_i \rightarrow H_i}((v,u)),$ where $\Pi_{G_i \rightarrow H_i}$ denotes the spanner embedding. 

    By our degree-control algorithm, we have that the number of vertices in $G_0$, denoted by $n_0 = \tilde{O}(m)$, and the edge count $m_0 = \tilde{O}(m)$. By \Cref{lm:StaticCoreGraphs} and the constant-factor increase due to degree-control, for any graphs $G_i, H_i$, the vertex count $n_i$ is at most $\tilde{O}(1) \cdot n_{i-1} / k$ and thus by induction, for $i > 0$, we have $n_i \leq m/k^i \cdot \tilde{O}(1)^i \leq m$ for reasonably large $m$ by the choice of $k$. Thus, by \Cref{thm:StaticSpanner}, the edge count $m_i$ of $H_i$ is at most $n_i \gamma \leq m \gamma$. All graphs have maximum degree $\Delta = \tilde{O}(\gamma)$.
  
    We note that the number of forest cycles is at most $m_i$ for each forest, so $\eta_i \cdot m_i = \tilde{O}(n\gamma)$ in total. By the crude approximation Lemma \ref{lem:brute_force} at a slight loss in quality of the approximate min-ratio cycle it suffices to check only $\tilde{O}(k\cdot \Delta^2)$ cycles per local ball neighborhood $B_{H_i}^{\mathcal{N}}(v,A)$ and thus over all $v \in V(H_{i-1})$ at most $\tilde{O}(n_i \cdot k\cdot \Delta^2) = \tilde{O}(m k \Delta^2)$ many cycles. The number of spanner cycles is again bounded by $m$.
    For the level $\Lambda$, we note that the number of vertices in $G_{\Lambda}$ is at most $\tilde{O}(m)/k^{\Lambda} \cdot \tilde{O}(1)^{\Lambda} = \tilde{O}(1)^{\Lambda} \leq \tilde{O}(1)^{\log^{1/21}m}$ which is smaller than $k$ for reasonably large $n$.  Therefore, we can use again the crude approximation Lemma \ref{lem:brute_force} to include those cycles. Over all levels, we collect at most $\tilde{O}(\Lambda \cdot (n\cdot k \cdot \Delta^2 + n\gamma)).$

    For the approximation, we have from the Vertex Sparsifier Dichotomy \ref{Cor:ASZ-Trichotomy} and the \ref{thm:SpannerDichotomy} that $\left(1/(20 \alpha) \cdot 1/(2\delta)\right)^{\Lambda} = (1/(40\alpha\delta))^{\Lambda}$-approximate min-ratio cycle is found either 1) in a local neighborhood $B_{H_i}^{\mathcal{N}}(v,A)$, 2) in forest cycle at some level $i$, 3) in some spanner cycle at some level $i$, or 4) is still contained in $G_{\Lambda}$. In case 2) and 3) the cycle is automatically extracted, otherwise, we use the procedure from Lemma \ref{lem:brute_force} to extract an approximate min-ratio cycle at a further loss in approximation of at most $1/\tilde{O}(k)$, since all local neighborhoods are of size $\tilde{O}(k)$ and $G_{\Lambda}$ contains at most $\tilde{O}(1)^{\Lambda} \leq k$ vertices for reasonably large $m$.

    The runtime bound follows from the runtime guarantees in Lemma \ref{lem:brute_force}, the fact that the spanner cycles are outputted explicitly and can thus be checked in the size of the support (see  \Cref{thm:StaticSpanner}) and the fact that each fundamental cycle with a forest $F_i$ can be evaluated in $\tilde{O}(k)$ time since each component of $F_i$ is of size at most $\tilde{O}(k)$ (see \Cref{lm:StaticCoreGraphs}).
\end{proof}

A simple calculation can now be used to show that we can compute a $m^{o(1)}$-approximate min-ratio achieved by the best cycle in almost-linear time $m^{1+o(1)}$. We abstain from doing a detailed analysis here since we are obtaining such a bound in the next section even for the dynamic min-ratio problem.

\section{Dynamic Min-Ratio Cycle via Dynamic Distance Oracles}\label{sec:Dynamization}

In \cite{KMPG24}, it is already described how to maintain distance oracles as described above. Here, we sketch their approach and show how it can be augmented straightforwardly to maintain the min-ratio cycle throughout a sequence of updates.

\subsection{Dynamic Vertex Sparsifier}\label{sec:DynamicVertexSparsifier}

The main objective is to describe how we maintain the gradient structure of one level of the sparsifier hierarchy under edge updates. We point out that the maintenance of the length structure of the entire vertex sparsifier hierarchy with respect to edge insertions, deletions, and vertex splits is already discussed in detail in \cite{KMPG24}. We first state a detailed summary of the result of \cite{KMPG24} and then give a discussion on how to augment it.

\paragraph{Maintenance of the \cite{KMPG24} Distance Oracles.} In order to implement the maintenance of the gradient structure in the vertex sparsifier, we recall that \cite{KMPG24} gives a method to maintain the vertex sparsifier of Lemma \ref{lm:StaticCoreGraphs}, the associated core graph sum from Definition \ref{def:StaticCoreGraphSum} and a spanner on the core graph sum as in \Cref{thm:StaticSpanner}. The Lemma below is an extended version of Lemma 3.1, in \cite{KMPG24}, where we make vertex sparsifier and core graph sums explicit. The details are not hard to verify from the proof in \cite{KMPG24}. We note that \cite{KMPG24} proves the Lemma below for graphs undergoing edge deletions only, but later describes how edge insertions and vertex splits can be emulated rather straightforwardly. We directly state the lemma obtained via these emulations here to avoid clutter.

\begin{lemma}[Dynamic Forests, Core Graphs and Spanner, c.f. Lemma 3.1 \cite{KMPG24}]\label{lm:dynamicforest}
    Given a dynamic graph $G$ where $\Delta$ bounds the maximum degree of all $G(t)$ and parameter $1 \le L \le o\left(\log^{1/6} m / \log\log m \right)$. There is an algorithm that maintains for
    some $\gamma = \exp(O(\log^{2/3} m \cdot \log\log m))$ and $\kappa = \tilde{O}(m^{4/K}\gamma^4)$ until time horizon $T = n/k$:
    \begin{itemize}
        \item an incremental $\beta = \tilde{O}(1) \cdot \Delta^2 \cdot k$ shattering pivot set $A$ of the current graph $G$, 
        \item $\eta = \tilde{O}(1)$ decremental forests $F_1, \ldots F_{\eta}$ of $G$, dynamic core graphs $\mathcal{C}(G,F_1), \dots, \mathcal{C}(G,F_\eta)$ and constant stretch bounds $\widetilde{\str}_i(e)$ for all $e \in E, i \in 1 \leq i \leq \eta$,
        \item $\gamma^{O(L)}$-spanner $H$ of the core graph sum $G' = \mathcal{C}(G, F_1, \dots, F_{\eta})$ with at most $n\gamma$ edges along with an explicitly spanner embedding $\Pi_{G' \mapsto H}$. $H$ at all times has maximum degree $\Delta^2 \kappa$.
    \end{itemize}
    It ensures that over the time horizon $T$, we have 
    \begin{enumerate}
        \item for all $i,$ each tree of $F_i$ has at most $\tilde{O}(k)$ many vertices and the set of roots of $F_i$ has size at most $\tilde{O}(n/k)$ and contains all vertices in $A$,
        \item \label{prop:allASZpathsCovered} for $\tilde{\mathcal{P}}$ being the ASZ path set for current graph $G$ and shattering set $A$ (as defined in \Cref{def:asz_path}), every path $P \in \tilde{\mathcal{P}}$ is preserved, i.e. $\forall t: \forall P \in \tilde{\mathcal{P}}: \exists i: \widetilde{\str}_{F_i(t)}(P) \leq \alpha = \tilde{O}(1).$
        \item \label{prop:Recourse} At time $t$ the number of updates to all forests $F_i$, core graphs $\mathcal{C}(G, F_i)$ and the spanner $H$ over all previous rounds is at most $\kappa \Delta^2 \cdot t$.
    \item The embedding $\Pi_{G' \to H}$ where $G' = \mathcal{C}(G, F_1, \dots, F_{\eta})$ satisfies 
    \begin{itemize}
        \item $\ell(\Pi_{G' \to H}((u,v)) \leq \gamma^{O(L)} \dist_{G'}(u,v) $,
        \item $\operatorname{card}(\Pi_{G' \to H}) \le \gamma^{O(L)}$,
        \item $\vcong(\Pi_{G' \to H}) \le \gamma^{O(L^2)} \Delta^2$.
    \end{itemize}
    \end{enumerate}
    The algorithm is deterministic, takes $O(m k^4 + m \Delta^2 + m \gamma)$ time for initialization, and  processes every update with worst-case time $O(k^4 + \Delta^2 \kappa \gamma^{O(L^2)} k^2)$. 
\end{lemma}

We want to highlight the remarkable fact presented in Property \ref{prop:allASZpathsCovered}, that the ASZ-paths of any current graph $G$ are preserved. This is possible by a brute-force operation that is central to the approach of \cite{KMPG24}: they show that the collection of paths in the input graph that can become an ASZ path at any point can be computed efficiently and has size $\tilde{O}(m \cdot \poly(k))$. Thus, they can preempt adversarial attacks already at initialization.

Another crucial point can be made about Property \ref{prop:Recourse}. While it seems that the recourse of the core graph and spanner are similar, this is somewhat deceiving as the recourse bounds the number of updates to these graphs. However, a vertex split has update complexity scaling in the degree of the affected vertex. Thus, the core graph sum has much higher update complexity (it could affect $\eta\beta$ edges), while $H$ has much smaller update complexity of $\kappa \cdot \Delta^2$. For the hierarchy, it is crucial that the recourse remains independent of $k$ so that $k$ can be vastly larger.

\paragraph{Maintaining Gradients.} Before we can discuss how to dynamically locate approximate min-ratio cycles efficiently, we have to discuss how to maintain the gradients. We first remark that updating the gradient structure of the spanner is trivial since the edges of the spanner are just edges of the original graph and since there are only $\eta = \tilde{O}(1)$ core graphs, we make at most $\eta$ changes. So there is no recourse overhead in performing such updates. 

Updating of the gradient structure of the vertex sparsifier is slightly more tricky: if a new root is chosen, then all edges incident to the new tree have to adapt their gradient which might result in high recourse. The key ingredient to reduce the recourse, introduced in previous mincost flow papers \cite{maxflow, BCKLPGSS23, incrflow}, are gradient shifts induced by vertex potentials. We recall that the gradient inner product of a circulation remains unchanged under a vertex-potential induced gradient shift.

\begin{fact}\label{fact:GradientShift}
    Let $\cc \in \R^E: \BB \cc = \veczero$ and $\gg \in \R^E, \pp \in \R^V,$ then $\gg^\top \cc = (\gg + \BB^\top \pp)^\top \cc.$ 
\end{fact}

We use this fact to zero out the gradient on the forest edges in the graph (see Definition \ref{def:StaticCoreGraphSum}). It is easy to compute for each forest $F_i$ of \Cref{lm:dynamicforest} a vertex-potential induced gradient $\gg_i = \BB^\top \pp_i$ such that $(\gg + \gg_i)|_{E(F_i)} = \veczero$ using a greedy algorithm. This is important to manage the gradient recourse for the core graph, because it ensures that changes to the forest do not affect the gradients of mapped edges. More precisely, since forests $F_i$ from  \Cref{lm:dynamicforest} are decremental, the initial potential $\pp_i$ still cancels out the gradients on the forest edges. Thus, no edge in the core graph needs to change its gradient even after the tree that it is incident to is split.  It also simplifies the definition of the core graph gradient which becomes $\gg_{\mathcal{C}_i}(\Pi_{\mathcal{C}_i}(e)) = \gg(e) + \vecone_e^\top \BB^\top \pp_i$ for any edge $e = (u,v) \in G$.

We recall the construction of the core graph sum. The core graph sum is obtained by taking $\eta$ copies of the original graph $G$ and connecting a star to all copies of a vertex $a \in A$. Denote this precursor graph by $G^{+}$. Then, by contracting the tree components of all forests, we obtain the core graph sum. Moreover, we recall from Definition \ref{def:L1CoreGraph} that the gradients for every core graph $\mathcal{C}_i$ in the core graph sum is given by $\gg_{\mathcal{C}_i}(\Pi_{\mathcal{C}_i}(e)) \defeq \sum_{e' \in \pproj_{F_i}(e)} \gg(e')$.
Let $\pp$ be the vertex potential on the core graph sum that is zero on all star centers and equal to $\pp_i$ on the i-th graph copy $G_i$. We point out that star edges that were previously assign $0$ gradient, may now carry gradients under the vertex potential. However, \Cref{fact:GradientShift} immediately tells us that we do not need to worry about using these edges to go between vertex copies of the same vertex in the core graph sum, they do not contribute to the gradient. Each circulation in the core graph sum thus maps back to a circulation in $G$ of the same gradient, as desired.

Finally, for updates, we perform the following actions: Given an edge insertion $e = (u,v)$, according to Lemma \ref{lm:dynamicforest}, we have that $u,v \in A$, and we simply set the edge gradient $\gg_{\mathcal{C}}(e)$ to $\gg(e) - \pp(v) + \pp(u)$. If the update is an isolated vertex $v$ insertion, we set its potential $\pp(v) = 0$. If it is a deletion, we can simply delete the edge and its gradient, and if it is a vertex split of a vertex $v$, we set the potential of both copies of $v$ to be $\pp(v)$. If we perform simply a change of gradient of an edge, we model this by deleting the edge with the old gradient and reinserting it with the new gradient. When a completely new vertex $v$ is added to the core graph sum through a vertex split, we split it further to obtain the star vertex again with the same potential. 

Our discussion immediately yields the following Lemma.

\begin{lemma}[Dynamic Gradient Maintenance]\label{lma:efficientGradientMaintenance}
    Given a dynamic core graph sum $\mathcal{C}(G, F_1, \dots, F_{\eta})$ of the underlying dynamic graph $G$, as in Lemma \ref{lm:dynamicforest}, we can maintain its gradient structure as defined in Definitions \ref{def:L1CoreGraph} and  \ref{def:StaticCoreGraphSum}. Per update to the core graph, we need to adapt the gradient of at most one edge in the core graph and we can do so in time $O(1)$.
\end{lemma}

\subsection{Dynamic Vertex Sparsifier Hierarchy}\label{sec:DynamicVertexSparsifierHierarchy}

To maintain the whole hierarchy consisting of alternate edge and vertex sparsification, \cite{KMPG24} resorts to a standard rebuilding scheme. Both the vertex and edge sparsifiers have sub-polynomial recourse $\tau = \kappa \Delta^2_i$ where $\Delta_i$ is the maximum degree of the level-$i$ graph that we apply the theorem to. Note that $\tau$ is independent of the size reduction factor $k$. Thus, by setting the size reduction significantly lower, we obtain a hierarchy with $\log_k m \ll \log_\tau m$ levels which yields subpolynomial recourse even at the highest level of the hierarchy. Standard rebuilding can further ensure that the vertex count remains reduced by factor almost $k$ on every level. We refer the reader to \cite{KMPG24} for a thorough treatment of the rebuilding scheme and only summarize their result here.

\begin{theorem}[see \cite{KMPG24}, Claim 4.6]\label{thm:maintainDynHierarchy}
Given dynamic input graph $G$ that has at any time maximum degree $3$. There is a deterministic algorithm, that maintains, for $\Lambda = \log^{1/21} m$ and $\chi = e^{O(\log^{20/21} m \log\log m)}$, the graphs $G_0 =G$ and $G_i$ being the maintained sparsifier from \Cref{lm:dynamicforest} on input $G_{i-1}$ with size reduction $k = m^{1/\Lambda}$ re-initialized at every time $t$ that is divisible by $u_i = m^{1-(i+1)/\Lambda}$. It thereby ensures that for each level $0 \leq i \leq \Lambda$, the number of vertices and edges in $G_i$ is at most $2m^{1-i/\Lambda} \cdot \chi$ at any time, its maximum degree is at most $\chi$. The algorithm takes initialization time $m \chi$ and thereafter update time $\chi$.
\end{theorem}

\subsection{Dynamic Min-Ratio Cycle Maintenance}

We argue next that we can maintain the approximate min-ratio cycle of the dynamic graph $G$ dynamically. We use $\chi$ often as a loose upper bound in the calculations below. While many calculations can be tightened, we believe that this would not result in substantially faster update times overall.

\begin{theorem}\label{thm:dynReportMinRatioValue}
Given dynamic input graph $G$ that has at any time maximum degree $3$. There is a deterministic algorithm, that can report after each update to $G$ an $\chi$-approximate min-ratio. More precisely, it explicitly maintains an ordered linked list of candidate cycles of size $\tilde{O}(m \cdot k \cdot \gamma^2 )$ such that some cycle $C$ of these candidate cycles is an approximate min-ratio cycle, i.e.
\[ \frac{\gg^\top \vecone_{C}}{\| \LL \vecone_C \|_1} \leq \frac{1}{\tilde{O}(k)}\left(\frac{1}{40 \alpha \cdot \delta}\right)^{\Lambda+1} \min_{\BB \cc = 0} \frac{\gg^\top \cc}{\|\LL \cc\|_1}. \]
Each cycle is supported on some graph $G_\ell$ as described in the hierarchy in \Cref{thm:maintainDynHierarchy} and of length at most $\zeta = \tilde{O}(\max\{k, \gamma^{O(L)}\}) = \tilde{O}(\chi)$. The algorithm takes initialization time $O(m\chi^5)$ and update time $O(\chi^5)$.
\end{theorem}
\begin{proof}
From our discussion of gradient maintenance on spanners and \Cref{lma:efficientGradientMaintenance}, we can augment the hierarchy from \Cref{thm:maintainDynHierarchy} to also maintain the gradients adjusted by vertex potentials at only constant multiplicative overhead.

Since  \Cref{thm:maintainDynHierarchy} maintains the hierarchy using \Cref{lm:dynamicforest}, it is straightforward to check that the Vertex Sparsifier Dichotomy \ref{Cor:ASZ-Trichotomy} and the Spanner Dichotomoy \ref{thm:SpannerDichotomy} still apply to each level. Thus, again, we have at any time an $\chi$-approximate min-ratio cycle 1) in a local neighborhood $B_{G_i}^{\mathcal{N}}(v,A)$, 2) in forest cycle at some level $i$, 3) in some spanner cycle at some level $i$, or 4) is still contained in $G_{\Lambda}$. 

For 1), we follow the same strategy as in the static setting and re-compute all neighorhoods that underwent changes. Any update (edge insertion/deletion, vertex splits, and isolated vertex insertions) to a graph $G_i$ affects only $\chi$ endpoints. Since the incremental vertex $A$ set was $\tilde{O}(\chi^3)$-shattering at initialization, and the shattering property is maintained under incremental updates to $A$, and thus the clusters have size at most $\tilde{O}(\chi^3)$. Therefore, these endpoints can be in at most $\tilde{O}(\chi^4)$ many local ball neighborhoods $B^{\mathcal{N}}(v,A)$. We can recompute the cycles associated with them using Lemma \ref{lem:brute_force} in time $\tilde{O}(k \cdot \Delta^2)$ and change their entries in the linked list. 

For 2), we note that since the forests on $G_i$ are decremental for level $i$ between times divisible by $u_i$, and since $G_i$ has at most $2m^{1-i/\Lambda} \cdot \chi$ edges at any time, we can check all tree cycles efficiently: at re-initialization we find all tree cycles and since the forests are $\tilde{O}(\chi^3)$-shattering, each cycle can be checked in time $\tilde{O}(\chi^3)$. Thereafter, we only have to remove cycles once the forests decompose. Once removed, cycles are then gone until the data structure is re-initialized. Only when an edge is inserted, we might have to add a new tree cycle, however, again this yields at most one new cycle per forest which can be evaluated in time $\tilde{O}(\chi^3)$.

For 3) we have that the spanner embeddings are maintained explicitly by the data structure. Further, by \Cref{lm:dynamicforest}, each cycle has cardinality at most $\chi$ which implies that we can check after each change to an embedding the affected cycles in time $O(\chi)$.

For 4), we have that $G_{\Lambda}$ remains of size at most $2\chi$ and thus we can use the procedure from Lemma \ref{lem:brute_force} at any time to find an approximate min-ratio cycle. 

Noticing that every update to a graph $G_i$ or a spanner embedding takes time at least $O(1)$ in the algorithm from \Cref{thm:maintainDynHierarchy} to write down the explicit update, and since each such update causes at most additional amortized time $\tilde{O}(\chi^4)$ by our proposed update strategy, yields the theorem.
\end{proof}

\section{A Simple Primal Min-Cost Flow Solver} \label{sec:solver}

In this section, we discuss how to solve the (incremental) min-cost flow problem using our simplified data structure to maintain min-ratio cycles. The main goal of this section is to use our new dynamic min-ratio cycle to implement the following data structure.

\begin{definition}[Solver, see Definition 3.7 in \cite{incrflow}]\label{def:solver}
We call a data structure that is initialized with a graph \(G=(V,E)\) and lengths \(\ll\in \mathbb{R}_{\ge 0}^{E}\), gradients \(\gg\in \mathbb{R}^{E}\), costs \(\cc\in \mathbb{R}^{E}\), a flow \(\ff\in \mathbb{R}^{E}\) routing demand \(\dd\), and a quality parameter \(q>0\), a step-size parameter \(\Gamma>0\) and an accuracy parameter \(\epsilon>0\), a \(\gamma_{\mathrm{approx}}\)-approximate min-ratio cycle solver if it
(implicitly) maintains a flow vector \(\ff\) such that \(\ff\)
routes demand \(\dd\) throughout and supports the following operations.

\begin{itemize}
    \item \textsc{ApplyCycle}(): One of the following happens.
    \begin{itemize}
        \item Either the data structure finds a circulation
        \(\Delta\in \mathbb{R}^{E}\) such that
        \[
        \frac{\gg^\top \Delta}{\|\LL\Delta\|_1}\le -\qq
        \qquad\text{and}\qquad
        |\gg^\top \Delta|=\Gamma .
        \]
        In that case, it updates $f \leftarrow \ff+\Delta$ and it returns a set of edges \(E'\subseteq E\)
        alongside the maintained flow values \(\ff(e')\) for
        \(e'\in E'\).

        For every edge \(e\), between the times that it is in
        \(E'\) during calls to \textsc{ApplyCycle}(),
        the value of \(l(e)f(e)\) does not change by more than
        \(\epsilon\).

        \item Or it certifies that there is no circulation
        \(\Delta\in \mathbb{R}^{E}\) such that
        \[
        \frac{\gg^\top \Delta}{\|\LL\Delta\|_1}
        \le -\frac{\qq}{\gamma_{\mathrm{approx}}}
        \]
        for some parameter \(\gamma_{\mathrm{approx}}\).
    \end{itemize}

    \item \textsc{UpdateEdge}\((e,\ll,\gg)\):
    Updates the length and gradient of an edge \(e\)
    that was returned by the last call to
    \textsc{ApplyCycle}().

    \item \textsc{InsertEdge}\((e,\ll,\gg,\cc)\):
    Adds edge \(e\) to \(G\) with length \(\ll\),
    gradient \(\gg\), cost \(\cc\).
    The flow \(\ff(e)\) is initialized to \(0\).

    \item \textsc{ReturnCost}():
    Returns the flow cost \(\cc^\top \ff\).

    \item \textsc{ReturnFlow}():
    Explicitly returns the currently maintained flow \(\ff\).
\end{itemize}

The sum of the sizes \(|E'|\) of the returned sets by
\(t\) calls to \textsc{ApplyCycle}() is at most $\frac{t\Gamma}{\epsilon q}.$
\end{definition}

As we show, we can obtain an efficient implementation by slightly augmenting our data structures. 

\begin{theorem}\label{thm:mainDataStuctureFlow}
There is a data structure as in Definition \ref{def:solver} that is a \(\chi\)-approximate min-ratio cycle solver. It implement the operations such that:
\begin{enumerate}
    \item The initialization takes time $m\chi. $
    \item The operation \textsc{ApplyCycle}()
    takes amortized time $(|E'|+1)\cdot \chi$ when \(|E'|\) edges are returned.
    \item The operations
    \textsc{UpdateEdge}()/
    \textsc{InsertEdge}()/
    \textsc{ReturnCost}()
    take amortized time $\chi.$
    \item The operation
    \textsc{ReturnFlow}()
    takes amortized time $m \chi.$
\end{enumerate}
\end{theorem}

In \cite{BLS23}, it was observed that the $\ell_1$-IPM from \cite{maxflow} can be augmented to give the following reduction from incremental min-cost flow to the above data structure problem. Combining the result below with \Cref{thm:mainDataStuctureFlow} immediately yields our main result \Cref{thm:mainIntro}. 

\begin{theorem}[Incremental threshold min-cost flow interior point method, see Claim 8.11 \cite{incrflow}]
There is an algorithm that given an incremental directed graph $G=(V,E)$ with polynomially bounded edge capacities \(u\) and costs \(c\), a threshold $F$,
a demand vector \(\dd \perp \mathbf{1}\), and access to a
\(\gamma_{\mathrm{approx}}\) min-ratio cycle solver data structure
\(\mathcal{D}\) maintains the solution to the thresholded min-cost flow problem (see \Cref{def:mincostFlow}) using total time $\widetilde{O}\!\left(
        m \gamma_{\mathrm{approx}}^{O(1)} \right)$ and using:
\begin{itemize}
    \item $\widetilde{O}\!\left(
        m \gamma_{\mathrm{approx}}^{O(1)}
    \right)$ calls to $\mathcal{D}.\textsc{ApplyCycle}()/
    \textsc{UpdateEdge}()/
    \textsc{InsertEdge}()/
    \textsc{ReturnCost}()$, and
    \item $\widetilde{O}\!\left(
        \gamma_{\mathrm{approx}}^{O(1)}
    \right).$ calls to $\mathcal{D}.\textsc{ReturnFlow}()$.
\end{itemize}
As long as the min-cost flow has cost at most $F$, the algorithm can also return a feasible flow of cost at most $F$ in additional time $\tilde{O}(m)$ and a single call to $\mathcal{D}.\textsc{ReturnFlow}()$.
\end{theorem}

\paragraph{Flow Solver Data Structure via Dynamic Min-Ratio Cycle}

It remains to describe how we implement the flow solver data structure claimed in \Cref{thm:mainDataStuctureFlow}. In order to give the implementation details, we finally have to discuss how to not only return an approximate min-ratio but a succinct representation of the approximate min-ratio cycle that allows us to apply the flow efficiently. Fortunately, \cite{KMPG24} already suggests a succinct encoding: it shows that there is an explicit low-recourse forest $F$ such that each edge in a graph $G_i$ in the dynamic hierarchy in \Cref{thm:maintainDynHierarchy} can be mapped to its corresponding path in $G$ (such a path can be found by recursively unpacking the ASZ paths). This allows us to use dynamic tree data structures to apply the flow efficiently and find edges whose lengths times flow product has increased significantly. The only catch is that $F$ is not a subgraph of $G$ but rather the graph obtained from $G$ by making subpolynomially many copies of each vertex and edge. We next describe the precise definitions and state the result obtained in \cite{KMPG24}. Finally, we describe how to use a dynamic tree data structure to prove \Cref{thm:mainDataStuctureFlow}.

\paragraph{Embedding the Graph Hierarchy into a Forest.} Let us first state the exact statement that can be obtained straightforwardly from using induction over the levels on Lemma 3.21 in \Cref{thm:mainDataStuctureFlow} (see also Section 4.2 of the paper were a very similar such induction is used).

\begin{lemma}\label{lma:embeddingLemma}
Given a dynamic graph $G$ of maximum degree $3$, along with the dynamic ASZ hierarchy of dynamic vertex sparsifiers $G_1,G_2, \ldots, G_{\Lambda}$ where each $G_\ell$ is a graph over vertices in $V \times [\chi]$. The algorithm from \cite{KMPG24} can explicitly maintain a dynamic forest $F$ over vertex set $V \times [\chi]$ such that:
\begin{itemize}
    \item for each edge $e = (u \times i) (v \times j) \in F$, we have $uv \in E, i,j \in [\chi]$ of length $\ell(e)$ and with gradient $\gg(e)$, and 
    \item each edge $f = (u \times i) (v \times j) $ from some graph $G_\ell$ can be mapped to $2\Lambda+1$ (possibly empty), path segments $P_1, P_2, \ldots, P_{2\Lambda+1}$ where each $P_s$ is a $(x_s \times i_s) (x_{s+1} \times j_s)$-path in $F$ where $x_1 = u$ and $x_{\Lambda +1} = v$. We further have that  $\gg_{\ell}(f) = \sum_s \gg(P_s)$ and $\ell_\ell(f) \leq \sum_s \ell(P_s)$.
\end{itemize}
The additional cost of maintaining $F$ and a map from edges $f$ in graphs $G_\ell$ to the endpoints of the $2\Lambda+1$ path segments in $F$ denoted by $\Pi$ incurs at most a constant multiplicative factor in the runtime.
\end{lemma}

To get such the embedding described in \Cref{lma:embeddingLemma}, \cite{KMPG24} makes the following observation: as can be seen from \Cref{def:StaticCoreGraphSum}, each edge $\Pi(f)$ in a graph at level $G_{\ell+1}$ is obtained from the projection of edge $f$ in the graph $G_{\ell}$ onto one of at most $\eta = \tilde{O}(1)$ forests $F_1, F_2, \ldots, F_{\eta}$. \cite{KMPG24} then shows that they can maintain the forests $F_1, F_2, \ldots, F_{\eta}$ explicitly using only $\eta$ copies of $G_{\ell}$. Each edge $\Pi(f)$ can thus be mapped down to two segments in forests $F_1, F_2, \ldots, F_{\eta}$ and the edge $f$ itself. Applying this process recursively, one only requires about $\eta^{\Lambda}$ many copies of each edge in the original graph $G$. The key insight in \cite{KMPG24} is that one can maintain the mapping of each $F_1, F_2, \ldots, F_{\eta}$ to $G$ explicitly. Thus, the map of $\Pi(f)$ to path segments in $F$ consists of the projection of the endpoints of the edge $f$ which is fully given in $F$, and the edge $f$ itself at the next lower level. Thus, we only require $2\Lambda$ projection segments, and the original edge from $G$. It is not hard to see from the definitions of core graphs (see \Cref{def:L1CoreGraph}), that lengths can be overestimated since projected lengths are overestimates while gradients are exactly preserved. We refer the reader for additional details to \cite{KMPG24}.

\paragraph{Applying Min-Ratio Cycles Efficiently.} Finally, we can describe how to implement \Cref{thm:mainDataStuctureFlow}. Using standard degree splitting techniques, we can assume that w.l.o.g. the incremental input graph has maximum degree $3$ at any time. We then maintain the hierarchy as in \Cref{thm:maintainDynHierarchy}  on $G$ augmented as described in \Cref{thm:dynReportMinRatioValue} to locate a $\chi$-approximate min-ratio cycle $C$ in time $\tilde{O}(\chi^5)$ in some graph $G_\ell$ where it consists of at most $\tilde{O}(\chi)$ edges. We maintain a dynamic tree data structure over the forest $F$ maintained by \Cref{lma:embeddingLemma}. 

The operations $\textsc{UpdateEdge}()/\textsc{InsertEdge}()$ can be forwarded immediately as graph updates, where updates are simulated by a single deletion and re-insertion.

For operation $\textsc{ApplyCycle}()$, if a good circulation is found, by \Cref{thm:dynReportMinRatioValue}, we get a cycle consisting of at most $\chi$ edges in some graph $G_\ell$. After computing the right amount of flow $\Delta$ that we send along the cycle, $\Delta$ flow is applies to each path segment in $F$ obtained by the map of each of the cycle edges. Since flow along a path segment can be applied implicitly in $O(\log n)$ time by a dynamic tree data structure, the total time of such an operation is only $\chi^{O(1)}$. Dynamic trees can also return all edges $e$ that received more than a certain amount of flow. While this is exactly what we would hope for to locate which values $\ll(e')\ff(e')$ increased by an $\epsilon$, we are faced with the obstacle that each edge $e' \in G$ has up to $\chi^2$ copies in $F$. But by taking $\epsilon/\chi^2$ as a parameter for the sensibility, we can ensure that every edge reported in $G$ had at least $\epsilon/\chi^2$ new flow (measure by length) while if it isn't reported, it does not exceed the threshold $\epsilon$ since its $\chi^2$ copies each carry less than $\epsilon/\chi^2$ flow (measured by lengths).

The operation $\textsc{ReturnCost}()$ can be implemented since it can be computed within the $\textsc{ApplyCycle}()$ operation how the applied cycle changes the cost of the circulation. Thus, we can simply store this cost and return it in $\tilde{O}(1)$ time.

Finally, to return the flow $f$, we simply read off the flow values of all edges in $F$ and collect them edge-wise to obtain the flow on $G$. Reading the flow value of an edge via a dynamic tree structure takes $O(\log n)$ time and thus, we can implement this operation in time $O(m \chi^2 \log n)$, as desired.

\bibliographystyle{alpha}

\bibliography{refs}

\end{document}